\documentclass[a4paper,aps,prl,intlimits,twocolumn,10pt]{revtex4-1}
\usepackage{amsmath,amsthm,amssymb,amsthm,dsfont,mathrsfs,mathtools,nicefrac}
\usepackage[colorlinks]{hyperref}
\usepackage[usenames,dvipsnames,svgnames,table]{xcolor}
\usepackage[english]{babel}
\usepackage{tikz-cd}

\theoremstyle{plain}
\newtheorem*{Th}{Theorem}

\begin{document}
\title{Groups, Information Theory and Einstein's Likelihood Principle}
\author{Gabriele Sicuro}\email{sicuro@cbpf.br}
\affiliation{Centro Brasileiro de Pesquisas F\'isicas, Rua Dr.~Xavier Sigaud, 150, 22290-180, Rio de Janeiro, Brazil}
\author{Piergiulio Tempesta}\email{p.tempesta@fis.ucm.es}
\affiliation{Departamento de F\'{\i}sica Te\'{o}rica II (M\'{e}todos Matem\'{a}ticos de la f\'isica), Facultad de F\'{\i}sicas, Universidad Complutense de Madrid, 28040 -- Madrid, Spain, and Instituto de Ciencias Matem\'aticas (CSIC-UAM-UC3M-UCM), C/ Nicol\'as Cabrera, No 13--15, 28049 Madrid, Spain}
\date{\today}
\begin{abstract}
We propose a unifying picture where the notion of generalized entropy is related to information theory by means of a group-theoretical approach. The group structure comes from the requirement that an entropy be well defined with respect to the composition of independent systems, in the context of a recently proposed generalization of the Shannon--Khinchin axioms. We associate to each member of a large class of entropies a generalized information measure, satisfying the additivity property on a set of independent systems as a consequence of the underlying group law. At the same time, we also show that Einstein's likelihood function naturally emerges as a byproduct of our informational interpretation of (generally nonadditive) entropies. These results confirm the adequacy of composable entropies both in physical and social science contexts.

\end{abstract}
\maketitle

The study of the relations among  statistical mechanics, information theory and the notion of entropy is at the heart of the science of complexity, and in the last decades has been widely explored.
After the seminal works of \textcite{Shannon} and \textcite{khinchin1957} on the foundations of Information Theory, \textcite{jaynes1957} re-formulated Boltzmann--Gibbs (BG) statistical mechanics as a statistical inference theory, where all fundamental equations are consequences of the maximum entropy principle applied to the \textit{BG entropy} $S^{\textrm{BG}}$. Subsequently, \textcite{renyi1961measure,renyi1970probability} introduced a generalized measure of information, now called the \textit{R\'enyi entropy} $S_\alpha^\textrm{R}$, that depends on a real parameter $\alpha$ having the BG entropy as a particular case in the $\alpha\to 1$ limit. 
The $S_q^\textrm{T}$ entropy, introduced by \textcite{havrda1967,tsallis1988}, has been the prototype of the nonadditive entropies studied in the last decades \cite{borges1998,canosa2002,beck2003,hanel2011,tempesta2011group,tsallis2013}. These functionals are generalizations of the BG entropy; they depend on one or more parameters, in such a way that the BG entropy is recovered as a particular limit. Generalized entropies have been successfully adopted for the study of both classical and quantum systems. R\'enyi's entropy, for example, plays a central role in information theory and in the study of multifractality \cite{jizba2004world}; along with the von Neumann entropy, it has been also used extensively in the evaluation of the entanglement entropy of quantum systems \cite{calabrese2009,coser2014,calabrese2014,calabrese2015}.


The study of new entropic forms has led to a new flow of ideas regarding the old problem of the probabilistic versus the dynamical foundations of the notion of entropy. It is well known \cite{cohen2002} that Einstein's approach was very different with respect to the probabilistic methodology of Boltzmann (which eventually emerged as the predominant one). Indeed, Einstein argued that the probabilities of occupation of the various regions of the phase space associated with a physical systems cannot be postulated \textit{a priori}. Instead, only a knowledge of dynamics, obtained by solving the equations of motion, could provide this information. For this reason, in his theory of fluctuations, \textcite{Einstein} introduced the likelihood function $\mathcal{W}\propto \exp\left(S^\textrm{BG}\right)$ as a fundamental statistical quantity (for the sake of simplicity, here and in the following, we put $k_\textrm{B}\equiv 1$, $k_\textrm{B}$ being the Boltzmann constant). He observed that, by composing two independent systems $\mathcal{A}$ and $\mathcal{B}$, the fundamental relation
\begin{equation}
\mathcal{W}(\mathcal{A}\cup \mathcal{B})=\mathcal{W}(\mathcal{A})\mathcal{W}(\mathcal{B}) \label{ELP}
\end{equation}
holds. Equation \eqref{ELP} is epistemologically crucial: it expresses the fact that the physical description of system $\mathcal{A}$ does not depend on the physical description of system $\mathcal{B}$, and \textit{vice versa}. Moreover, it is related to the additivity requirement of the information content of independent systems, as will be explained below.

The likelihood function has a clear physical meaning. It provides the number of accessible configurations in the entire space of possible configurations. In many circumstances, this number is exponential in the size $N$ of the system; therefore we can write $\mathcal{W}\propto \exp\left(N\Sigma\right)$, where $\Sigma$ is an (adimensional) entropy density. \textcite{Parisi} introduced this quantity in the study of disordered systems, calling $\Sigma$ the complexity, or configurational entropy. An equivalent quantity is used in the study of random optimization problems by means of statistical physics techniques \cite{Zamponi}.

By following the analysis in \cite{tsallis2015boltzmann}, we shall call Eq.~\eqref{ELP} the Einstein's likelihood principle; it is a fundamental relation in Einstein's theory of fluctuations \citep{Einstein}. Remarkably, in the context of superstatistics \cite{beck2003}, \textcite{abe2007} generalized Einstein's principle for nonequilibrium systems in the presence of (quenched) temperature fluctuations and using conditional entropies.

The aim of this paper is to provide a novel approach that relates, in a unique framework, classical information theory with both the notion of generalized entropy and Einstein's likelihood principle. Precisely, we shall show that an intrinsic \textit{group--theoretical structure} is at the heart of the multiple connections among these foundational perspectives.

The physical root of this group structure relies in a recent generalization of the classical axiomatic formulation \cite{tempesta2014beyond,tempesta2015new}, originally proposed by Shannon and by Khinchin to characterize the BG entropy. The first three postulates, nowadays called the Shannon--Khinchin (SK) axioms \cite{Shannon,khinchin1957}, define some essential properties that any entropic functional $S$ should satisfy. Let us consider the set $\mathscr P$ of finite discrete probability distributions $P\in\mathscr P$, $P\coloneqq\{p_1,\dots,p_W\}$, $W\in\mathds N\setminus\{0\}$, $p_i\geq 0$, $\sum_{i=1}^Wp_i=1$. We can state the postulates as follows.
\begin{description}
\item[Continuity] The function $S(p_1,\dots,p_W)$ is continuous with respect to all its arguments.
\item[Maximum principle] The function $S(p_1,\dots,p_W)$ is maximized by the uniform distribution.
\item[Expansibility] Adding an event of zero probability does not change the value of the entropy, i.e., $S(p_1,\dots,p_W,0)\equiv S(p_1,\dots,p_W)$.
\end{description}
Also, a fourth axiom, i.e.~\textit{additivity} with respect to the composition of two systems, was required. Under these assumptions, Khinchin proved that the only admissible entropy turns out to be the BG entropy $S^\textrm{BG}[P]\coloneqq-\sum_{i=1}^Wp_i\ln p_i$.

The additivity property was thought to be a sufficient condition for the \textit{extensivity} of BG entropy, which in the formulation of Clausius is an essential requisite for thermodynamics. However, in the last decades it became evident that the two properties, i.e.~additivity and extensivity, are unrelated \cite{touchette2002}. Indeed, denoting by $W(N)$ the number of accessible states of a system with $N$ particles, the BG entropy is not extensive, \textit{on the uniform distribution}, if $W(N)\sim N^\alpha$ for a certain $\alpha\in\mathds R^+$. This scaling is not atypical and it appears often in the framework of complex systems (see e.g.~\cite{hanel2011b}). Recently, the non-extensivity of Boltzmann's entropy over a large class of probability distributions was proved. Surprisingly, R\'enyi's entropy can be extensive in the same contexts \cite{bergeron2016}. Additivity is therefore not an intrinsic physical requirement. At the same time, by renouncing the additivity postulate, new possibilities arise.



In the context of information theory, non--additive entropies provided useful tools, for example, in the study of quantum entanglement \cite{canosa2002,berta2010}. However, the lack of additivity appears to be an important flaw if we want to interpret generalized entropies as \textit{classical} measures of information \cite{renyi1970probability,jizba2004world}. In the following, we shall assume that each system is in a given macrostate, corresponding to a set of possible configurations, or microstates; a probability distribution function is defined over all allowed microstates. In this sense, each state is associated with a given probability distribution that, in our analysis, determines both the entropy and the information measure of the system. With a slight abuse of language, we will call the information measure obtained from this probability distribution as the information measure of the system itself. It is expected that, given two distributions associated with two independent systems (or with the results of two independent experiments), the corresponding total amount of information is nothing but the sum of the information measures obtained by each system (experiment). Moreover, any measurement or change of information in one of the two systems does not affect, nor is influenced by, any property of the other system, being the systems uncorrelated.

This property is certainly satisfied by Boltzmann's and R\'enyi's entropies; however, using non--additive entropic functionals \textit{tout--court} it is not possible to fulfill the requirement above. Therefore, to preserve the interpretation of entropy as a generalized information measure $I(\mathcal A)$ of a given system $\mathcal A$, we firstly postulate the general relation
\begin{equation}\label{Info1}
I(\mathcal A)=f(S(\mathcal A)),
\end{equation}
where the information is assumed to be a function $f$ of the entropy of the system only (hereafter we avoid to make explicit the dependence of $I$ on the specific probability distribution $P$ associated with a macrostate of $\mathcal{A}$).

Second, along the lines of \cite{tempesta2014beyond,tempesta2015new}, we discuss an axiomatic formulation of entropy, generalizing the fourth SK axiom and requiring, instead of additivity, the \textit{composability} property \cite{tsallis2009I}. In this settings, we postulate that the entropy of a composite system be a function of the entropy of the two components only. This simple requirement is, however, rich with consequences. In particular, a specific algebraic structure appears, in which the composition operation plays the role of group operation. In the following, we will show that, for a large class of entropies, an additive information measure having the form of Eq.~\eqref{Info1} is automatically inferred by the algebraic structure itself.
The composition process of two or more systems is formalized as follows:

\begin{description}
\item[Composability axiom] Given two statistically independent systems $\mathcal A$ and $\mathcal B$, each defined on a given probability distribution in $\mathscr P$, there exists a symmetric function $\Phi(x,y)$ such that
\begin{equation}\label{composability}
S(\mathcal A\cup\mathcal B)=\Phi\left(S(\mathcal A),S(\mathcal B)\right).
\end{equation}
Moreover, we require the associativity property
\begin{equation}
\Phi\left(\Phi(x,y),z\right)=\Phi\left(x,\Phi(y,z)\right)
\end{equation}
and the relation $\Phi(x,0)=x$.
\end{description}

In \cite{beck2009} axiomatic formulations of the notion of entropy, and in particular the composability property for non-independent systems have been discussed in a different framework.

We shall define \textit{admissible entropies} as those satisfying the first three SK axioms and the \textit{composability axiom}.
The  axiom contains fundamental requirements, for instance the existence of a zeroth law of thermodynamics. Also, when composing a system with another in a certainty state (zero entropy), the entropy of the composed system will remain unchanged. This natural generalization of the SK axiom allows an infinite number of admissible entropic forms \cite{tempesta2014beyond,tempesta2015new}.


The classification and study of the properties of all functions $\Phi(x,y)$ satisfying the composability axiom has been performed in the context of formal group theory \cite{hazewinkel1978formal}. This branch of algebraic topology was established in the second half of the 20th century, starting from the work of \textcite{bochner1946formal}. In particular, there exists a universal group law, the \textit{Lazard formal group}, essential to the discussion of composability.

More precisely, we shall consider a monotonic, strictly increasing function $G(t)$ admitting in $\mathds{R}$ a power series expansion of the form $G\left( t\right) =\sum_{k=0}^{\infty} a_k \frac{t^{k+1}}{k+1}$, where $\{a_{k}\}_{k\in\mathds{N}}$ is a suitable real sequence, with $a_{0}\neq 0$. This series, by means of the Lagrange inversion theorem, is invertible with respect to the composition (i.e., there exists a series $G^{-1}$ such that $G^{-1}(G(t))=t$). The function
\begin{equation}\label{Lazard}
\Phi(x,y)=G\left(G^{-1}(x)+G^{-1}(y)\right).
\end{equation}
defines a group law. One can prove that, given a formal group law $\Phi(x,y)$, there exists a specific power series $G(t)=a_0t+O(t^2)$ such that the law can be represented as in Eq.~\eqref{Lazard}
\footnote{This set of properties holds for group laws on zero--characteristic fields. For further details, see \cite{tempesta2015new}.}.

Consequently, if a functional $S$ satisfies the generalized SK axioms, then the associated composition law \eqref{composability} can be realized in terms of the law \eqref{Lazard} by means of a specific function $G$.  The rich algebraic structure underlying the composability axiom naturally allows us to classify all entropies possessing a given composition law. Also, we can generate new examples, by varying  the composition law adopted. For the trace-form family (i.e., the family of entropies having the structure $S[P]=\sum_ip_if(p_i)$ for a sufficiently regular $f$), a natural general form is \cite{tempesta2015new}
\begin{equation}
S_G[P]=\sum_{i=1}^Wp_iG\left(\ln\frac{1}{p_i}\right).\label{tf}
\end{equation}
For example, the celebrated Boltzmann entropy corresponds to the \textit{additive group}:
\begin{equation}
\Phi(x,y)=x+y\Rightarrow G(t)=kt,\quad k\in\mathds R^+,
\end{equation}
the constant $k$ being usually identified with the Boltzmann constant.

\noindent The \textit{multiplicative group}
\begin{equation}\label{multgr}\Phi(x,y)=x+y+(1-q)  xy\Rightarrow G(t)=\frac{e^{a(1-q) t}-1}{1-q}\end{equation}
leads to a generalization of the Tsallis entropy $S_q^\mathrm{T}$, defined for $q\in\mathds R$ and $a\in\mathds R^+$ as  \cite{tempesta2015new}
\begin{equation}
S_{a,q}[P]\coloneqq \sum_{i=1}^Wp_i\log_q\frac{1}{p_i^a},
\end{equation}
where $\log_q(x)\coloneqq\frac{x^{1-q}-1}{1-q}\xrightarrow{q\to 1}\ln x$, $x\in\mathds R^+$. Observe that
\begin{equation}
\lim_{a\to 1}S_{a,q}[P]=S^\mathrm{T}_q[P]\quad\text{and}\quad \lim_{q\to 1}S_{a,q}[P]=a S^\mathrm{BG}[P].
\end{equation}
Remarkably, the $S_{a,q}[P]$ entropy is the most general known trace-form entropy that is admissible according to the definition above, including the BG entropy and the Tsallis entropy as particular cases. All other entropies in the form \eqref{tf} satisfy the composability property on the uniform distribution. Dropping out the trace-form hypothesis, new entropic forms are allowed \cite{tempesta2015new}. For example, the R\'enyi entropy
\begin{equation}
S^\textrm{R}_\alpha[P]\coloneqq \frac{\ln\sum_{i=1}^Wp_i^\alpha}{1-\alpha}\xrightarrow{\alpha\to 1}S^\textrm{BG}[P]
\end{equation}
is an additive composable entropy.

Motivated by the previous discussion, we can propose now a notion of information measure that comes directly from the group-theoretical structure. To this regard, one should notice that the additivity property is lost for a very large class of composable entropies discussed above. However, we can overcome this difficulty by associating with each of them an information measure that is indeed additive on statistically independent systems. This measure is determined by the composition group itself, i.e., by the function $G$ appearing in Eq.~\eqref{Lazard}. Consequently, we propose the following main definition of a \textit{group-theoretical generalized information measure}.

Given a composable entropy $S_G$, with a group law of the form \eqref{Lazard}, the information measure associated with $S_G$ for any system $\mathcal{A}$ corresponding to a given probability distribution is defined to be
\begin{equation}
{I}_{G}(\mathcal{A})= G^{-1}(S_G(\mathcal{A})). \label{IF}
\end{equation}

In the specific case of an entropy of trace-form class \eqref{tf}, we recover for our functional the expression of the \textit{Kolmogorov-Nagumo mean} \cite{kolmogorov,nagumo}:
\begin{equation}
{I}_{G}(\mathcal{A})=G^{-1}\left(\sum_{i=1}^{W}p_i G\left(\ln\frac{1}{p_i}\right)\right).
\end{equation}

The information measure \eqref{IF}, however, is defined in the generic case of entropies that are composable and not only for trace-form entropies. We are ready now to present one of the main results of this paper.
\begin{Th}
Let $S_G$ be a composable entropy, with a group law defined by \eqref{Lazard} for a certain function $G$. Then for two statistically independent systems $\mathcal A$ and $\mathcal B$ we have
\begin{equation}{I}_{G}(\mathcal A\cup\mathcal B)={I}_{G}(\mathcal A)+{I}_{G}(\mathcal B).\end{equation}
Moreover, $I_G$ satisfies the following further properties.
\begin{description}
\item[Continuity] $ I_G$ is continuous with respect to its arguments;
\item[Maximum principle] $ I_G$ is maximized on the uniform distribution;
\item[Expansibility] The addition of a zero--probability event does not change the value of ${I}_G$.
\end{description}
\end{Th}
\begin{proof} Observe indeed that
\begin{multline}{I}_{G}(\mathcal A\cup\mathcal B)=\\=G^{-1}(S_G(\mathcal A\cup \mathcal B))=G^{-1}\left(\Phi(S_G(\mathcal A), S_G(\mathcal B))\right)\\ =G^{-1}\left\{G\left[G^{-1}(S_G(\mathcal A))+G^{-1}(S_G(\mathcal B))\right]\right\}\\
={I}_{G}(\mathcal A)+{I}_{G}(\mathcal B).\end{multline}
All the other properties of $ I_G$ derive from the properties of $S_G$ imposed by the generalized SK axioms and from the strict monotonicity of $G$.\end{proof}
Therefore all composable entropies possess an associated information measure which is additive and can be constructed directly through the function $G$. Observe that the strict monotonicity of $G$ (and therefore of $G^{-1}$) implies that $S_G(\mathcal A)<S_G(\mathcal B)\Rightarrow  I_G(\mathcal A)<  I_G(\mathcal B)$, coherently with the fact that  the entropic forms allowed by the axioms above do possess indeed an information content even if non--additive (see for example \cite{tsallis1998,canosa2002} for the study of the Tsallis entropic form as information functional and for applications).

Observe also that the R\'enyi entropy fits naturally in our scheme. Indeed, for $S^\textrm{R}_\alpha$, $G(t)=t$, then the associated information measure $ I^\textrm{R}_\alpha$ coincides with the entropic functional, i.e.
\begin{equation}
S^\textrm{R}_\alpha(\mathcal A)\mapsto  I^\textrm{R}_\alpha(\mathcal A)\equiv S^\textrm{R}_\alpha(\mathcal A).
\end{equation}
The previous identity holds also in the $\alpha\to 1$ limit, i.e., in the BG case. In other words, the BG entropy and the R\'enyi entropy are stable with respect to definition \eqref{IF}: the associated group-theoretical information measure coincides with the corresponding entropy. If we consider instead the $S_{a,q}$ entropy, we see that
\begin{equation}
S_{a,q}(\mathcal A)\mapsto I_{a,q}(\mathcal A)\equiv S^\textrm{R}_{1+a(q-1)}(\mathcal A).
\end{equation}
In particular, the information measure associated to $S_q^\mathrm{T}$ ($a\equiv 1$) is the R\'enyi entropic functional with parameter $q$. However, as explained above, our formalism goes beyond the requirement of a Kolmogorov--Nagumo structure and it holds, indeed, for all nontrace-form admissible entropies, having a more general form for the corresponding information functional. For a large family of nontrace-form entropies (for example, for the entropic functionals discussed in \cite{tempesta2015new}), the analytic expression of the information functional is of the type
\begin{equation}
S_G(\mathcal A)\mapsto  I_G(\mathcal A)=\sum_i c_i S^\textrm{R}_{\alpha_i}(\mathcal A),
\end{equation}
where the parameters $c_i$ and $\alpha_i$ depend on the parameters appearing in the entropy $S_G$.

Finally, the presented group-theoretical approach allows us to generalize easily Einstein's principle, and to connect it with information theory in a natural way. Given an entropy $S_G$, whose associated group law is Eq.~\eqref{Lazard}, we introduce the likelihood function
\begin{equation}\label{likelihood}
\mathcal W_G(\mathcal A)\coloneqq e^{ I_G(\mathcal A)}\equiv e_G^{S_G(\mathcal A)},
\end{equation}
where
\begin{equation}
e_G(x)\coloneqq \exp\left[G^{-1}(x)\right].
\end{equation}
The motivation for this definition is twofold. First, it relates Einstein's likelihood function directly with a group-theoretical information measure. Second, it generalizes Einstein's relations \cite{Einstein,tsallis2015boltzmann} in the case of independent systems for all composable entropies. Indeed, let $S$ be a composable entropy, whose group law is given by Eq.~\eqref{Lazard}. Then Einsten's likelihood principle \eqref{ELP} follows immediately from the additivity property of the information functional \eqref{IF}. In the case of the multiplicative group \eqref{multgr}, we recover the likelihood function recently introduced in \cite{tsallis2015boltzmann}. Once again, for generalized entropies that are composable only over the uniform distribution, just a weak formulation of the principle holds (needless to say, this situation is not very satisfactory).

In the light of the whole analysis of this paper, we conclude that the composability axiom allows a potentially fruitful interpretation of generalized entropies in information theory. Indeed, composable entropies both possess an information theoretical content and satisfy Einstein's principle, which is a crucial statement for any physical application of the notion of entropy. As a byproduct of our analysis, it emerges that also nontrace-form but composable entropies can play an important role in statistical mechanics.

\begin{acknowledgments}
\paragraph{Acknowledgments}We thank Francesco Toppan and Constantino Tsallis  for useful discussions. G.S.\@ acknowledges the financial support of the John Templeton Foundation. The research of P.T. has been partly supported by the project FIS2015-63966, MINECO, Spain, and by the ICMAT Severo Ochoa project SEV-2015-0554 (MINECO).
\end{acknowledgments}
\bibliography{Biblio.bib}
\end{document}